\documentclass[final,11pt]{article}
\usepackage[left=1in, top=1in, bottom=1in, right=1in]{geometry}
\usepackage{enumitem}

\usepackage{latexsym,amssymb,amsmath,amsthm,color}
\usepackage{graphicx}
\usepackage[colorlinks=true]{hyperref}
\hypersetup{urlcolor=blue, citecolor=red}
\usepackage{cite}
\usepackage[table]{xcolor}

\newcommand{\R}{\mathbb{R}}

\newcommand{\be}{\begin{equation}}
\newcommand{\ee}{\end{equation}}
\newcommand{\bea}{\begin{eqnarray}}
\newcommand{\eea}{\end{eqnarray}}

\newcommand{\E}[1]{\mathrm{E}\left\{{#1} \right\}}
\newcommand{\Ez}[1]{\mathrm{E}_z\left\{{#1} \right\}}
\newcommand{\mat}[1]{\mathbf{{#1}}}

\newtheorem{theorem}{Theorem}[section]
\newtheorem{lemma}{Lemma}[section]
\newtheorem{proposition}{Proposition}[section]
\newtheorem{remark}{Remark}[section]

\def\1{{\rm 1}}

\def\s1{^{\rm (1)}}

\newcommand{\V}{\mathbb{V}}

\renewcommand{\vec}[1]{{\mathchoice
                     {\mbox{\boldmath$\displaystyle{#1}$}}
                     {\mbox{\boldmath$\textstyle{#1}$}}
                     {\mbox{\boldmath$\scriptstyle{#1}$}}
                     {\mbox{\boldmath$\scriptscriptstyle{#1}$}}}}

\newcommand{\ip}[2]{\left\langle {#1}, {#2} \right\rangle}

\def\1{{\rm 1}}

\def\s1{^{\rm (1)}}

\newcommand{\Np}{{N_\text{p}}}

\def\longrightharpoonup{\relbar\joinrel\rightharpoonup}
\def\longleftharpoondown{\leftharpoondown\joinrel\relbar}
\def\longrightleftharpoons{
  \mathop{
    \vcenter{
      \hbox{
      \ooalign{
        \raise1pt\hbox{$\longrightharpoonup\joinrel$}\crcr
          \lower1pt\hbox{$\longleftharpoondown\joinrel$}
          }
      }
    }
  }
}

\usepackage[table]{xcolor}
\usepackage{booktabs} 
\usepackage{relsize}
\usepackage{bm}
\usepackage{mathrsfs}
\usepackage{amsmath,amssymb}

\usepackage[latin1]{inputenc}
\usepackage{tikz}
\usetikzlibrary{positioning,shapes,arrows,arrows.meta}
\usetikzlibrary{shapes.geometric}
\tikzstyle{decision} = [diamond, minimum width=2.5cm, minimum height=1cm, text centered, draw=black]
\tikzstyle{startstop} = [rectangle, draw, 
   text width=4.5em, text centered, rounded corners]
 \tikzstyle{process} = [rectangle, draw, 
   text width=9.0em, text centered]  
\tikzstyle{line} = [draw, -latex']
\tikzstyle{cloud} = [draw, ellipse, node distance=3cm,
   minimum height=2em]
   \tikzstyle{io} = [trapezium, trapezium left angle=70, trapezium right angle=110, minimum width=2cm, text width=9.5em,minimum height=1cm, text centered, draw, inner sep=0pt]

\usepackage{algorithm,algpseudocode,float}
\usepackage{lipsum}
\makeatletter

\newenvironment{breakablealgorithm}
  {
   \begin{center}
     \refstepcounter{algorithm}
     \hrule height.8pt depth0pt \kern2pt
     \renewcommand{\caption}[2][\relax]{
       {\raggedright\textbf{\ALG@name~\thealgorithm} ##2\par}%
       \ifx\relax##1\relax 
         \addcontentsline{loa}{algorithm}{\protect\numberline{\thealgorithm}##2}%
       \else 
         \addcontentsline{loa}{algorithm}{\protect\numberline{\thealgorithm}##1}%
       \fi
       \kern2pt\hrule\kern2pt
     }
  }{
     \kern2pt\hrule\relax
   \end{center}
  }
  \makeatother
\algdef{SE}[DOWHILE]{Do}{doWhile}{\algorithmicdo}[1]{\algorithmicwhile\ #1}%

\begin{document}

\thispagestyle{empty}
\begin{center}
\textsc{Active subspace-based dimension reduction for chemical kinetics applications with epistemic uncertainty}

\bigskip 
\bigskip 

Manav Vohra$^{1}$, Alen Alexanderian$^{2}$, Hayley Guy$^{2}$, Sankaran Mahadevan$^{1}$

\bigskip
\bigskip

\normalsize
$^1$Department of Civil and Environmental Engineering\\
Vanderbilt University\\
Nashville, TN 37235\\

\bigskip

$^2$Department of Mathematics\\
North Carolina State University\\
Raleigh, NC 27695\\

%
\end{center}
%
%
%




\section*{Abstract}
We focus on an efficient approach for quantification of uncertainty in complex chemical
reaction networks with a large number of uncertain parameters and input conditions.  Parameter
dimension reduction is accomplished by computing an~\emph{active subspace} that
predominantly captures the variability in the quantity of interest (QoI).  In
the present work, we compute the active subspace for a H$_2$/O$_2$ mechanism 
that involves 19 chemical reactions, using an efficient iterative strategy.  The
active subspace is first computed for a 19-parameter problem wherein only the
uncertainty in the pre-exponents of the individual reaction rates is
considered. This is followed by the analysis of a 36-dimensional case wherein
the activation energies and initial conditions are also considered uncertain.  In both cases, a
1-dimensional active subspace is observed to capture the uncertainty in the QoI,
 which indicates enormous
potential for efficient statistical analysis of complex chemical systems. In
addition, we explore links between active subspaces and global sensitivity
analysis, and exploit these links for identification of key
contributors to the variability in the model response.
\bigskip

\noindent \textbf{Keywords}: Chemical kinetics; epistemic uncertainty; active subspace; dimension reduction; surrogate

\clearpage

\section{Introduction}
\label{sec:intro}


Time evolution of a chemically reacting system is largely dependent upon rate
constants associated with individual reactions. The rate constants are
typically assumed to exhibit a certain correlation with temperature (e.g.,
Arrhenius-type). Hence, accurate specification of the rate-controlling
parameters is critical to the fidelity of simulations. However, in practical
applications, these parameters are either specified using expert knowledge or
estimated based on a regression fit to a set of sparse and noisy
data~\cite{Burnham:1987, Burnham:1988, Vohra:2011, Sarathy:2012}.
Intensive research efforts in recent years within the field of uncertainty quantification (UQ)
address the quantification and propagation of uncertainty in system models due to 
inadequate data, parametric uncertainty, and model errors~\cite{Vohra:2014, 
Vohra:2017, Morrison:2018, Hantouche:2018, Nannapaneni:2016, Sankararaman:2015,
Reagana:2003}. 

In complex mechanisms involving a large number of reactions, characterizing the
propagation of uncertainty from a large set of inputs to the model output is
challenging due to the associated computational effort.  A major focus of this
article is the implementation of a robust framework that aims to identify
\emph{important} directions in the input space that predominantly capture the
variability in the model output. These directions, which constitute the so called
\emph{active subspace}~\cite{Constantine:2015}, are the dominant eigenvectors
of a matrix derived from the gradient information of the model output with
respect to an input. The active subspace methodology thus focuses on reducing
the dimensionality of the problem, and hence the computational effort
associated with uncertainty propagation. The focus here is on input parameter
dimension reduction. This is different from 
techniques such at Computational
Singular Perturbation (CSP)~\cite{Lam85,
LamGoussis89,ValoraniCretaGoussisEtAl06,
DebusschereMarzoukNajmEtAl12,
SalloumAlexanderianLeMaitreEtAl12} that aim to reduce the
complexity of stiff chemical systems by filtering out the 
fast timescales from the system. The latter is done, for instance, using
the eigenvectors of the system Jacobian 
to decouple the fast and slow processes; see e.g.,~\cite{DebusschereMarzoukNajmEtAl12}.

The application problem considered in this work is the
H$_2$/O$_2$ reaction mechanism from~\cite{Yetter:1991}. This mechanism has gained
a lot of attention as a potential source of clean energy for
locomotive applications~\cite{Das:1996}, and more recently in fuel
 cells~\cite{Loges:2008,Cosnier:2016}. 
The mechanism involves 19
reactions including chain reactions, dissociation/recombination reactions, and
formation and consumption of intermediate species; see Table~\ref{tab:kinetics}. 
For each reaction, the reaction rate is assumed to follow an Arrhenius
correlation with temperature:
\be
k_i(T) = A_iT^{n_i}\exp(-E_{a,i}/RT), 
\label{eq:rate}
\ee
where $A_i$ is the pre-exponent, $n_i$ is the temperature
exponent, $E_{a,i}$ is the
activation energy corresponding to the $i^{th}$ reaction, and $R$ is the
universal gas constant.  The Arrhenius rate law in~\eqref{eq:rate} 
is often interpreted in a logarithmic form as follows:
\be
\log(k_i) = \log(A_i) + n_i\log(T) - E_{a,i}/RT. 
\label{eq:ratelog}
\ee
\begin{table}[htbp]
\renewcommand*{\arraystretch}{0.9}
\begin{center}
\begin{tabular}{llll}
\toprule
Reaction \#     & Reaction &&\\
\bottomrule
$\mathcal{R}_1$ & H + O$_2$          & $\rightleftharpoons$ & O + OH \\
$\mathcal{R}_2$ & O + H$_2$          & $\rightleftharpoons$ & H + OH \\
$\mathcal{R}_3$ & H$_2$ + OH         & $\rightleftharpoons$ & H$_2$O + H \\
$\mathcal{R}_4$ & OH + OH            & $\rightleftharpoons$ & O + H$_2$O \\
$\mathcal{R}_5$ & H$_2$ + M          & $\rightleftharpoons$ & H + H + M \\
$\mathcal{R}_6$ & O + O + M          & $\rightleftharpoons$ & O$_2$ + M \\
$\mathcal{R}_7$ & O + H + M          & $\rightleftharpoons$ & OH + M \\
$\mathcal{R}_8$ & H + OH +M          & $\rightleftharpoons$ & H$_2$O + M \\
$\mathcal{R}_9$ & H + O$_2$ + M      & $\rightleftharpoons$ & HO$_2$ + M \\
$\mathcal{R}_{10}$ & HO$_2$ + H      & $\rightleftharpoons$ & H$_2$ + O$_2$ \\
$\mathcal{R}_{11}$ & HO$_2$ + H      & $\rightleftharpoons$ & OH + OH \\
$\mathcal{R}_{12}$ & HO$_2$ + O      & $\rightleftharpoons$ & O$_2$ + OH \\
$\mathcal{R}_{13}$ & HO$_2$ + OH     & $\rightleftharpoons$ & H$_2$O + O$_2$ \\
$\mathcal{R}_{14}$ & HO$_2$ + HO$_2$ & $\rightleftharpoons$ & H$_2$O$_2$ + O$_2$ \\
$\mathcal{R}_{15}$ & H$_2$O$_2$ + M  & $\rightleftharpoons$ & OH + OH + M \\
$\mathcal{R}_{16}$ & H$_2$O$_2$ + H  & $\rightleftharpoons$ & H$_2$O + OH \\
$\mathcal{R}_{17}$ & H$_2$O$_2$ + H  & $\rightleftharpoons$ & HO$_2$ + H$_2$ \\
$\mathcal{R}_{18}$ & H$_2$O$_2$ + O  & $\rightleftharpoons$ & OH + HO$_2$ \\
$\mathcal{R}_{19}$ & H$_2$O$_2$ + OH & $\rightleftharpoons$ & HO$_2$ + H$_2$O \\
\bottomrule
\end{tabular}
\end{center}
\caption{Reaction mechanism for H$_2$/O$_2$ from~\cite{Yetter:1991}}.
\label{tab:kinetics}
\end{table}
The global reaction associated with the H$_2$/O$_2$ mechanism can
be considered as follows:
\be
2\text{H}_2 + \text{O}_2 \rightarrow 2\text{H}_2\text{O}.
\label{eq:global}
\ee 
The equivalence ratio ($\Phi$) is given as follows:
\be
\Phi = \frac{(M_{\text{H}_2}/M_{\text{O}_2})_\text{obs}}{(M_{\text{H}_2}/M_{\text{O}_2})_\text{st}},
\label{eq:phi}
\ee
where the numerator on the right-hand-side denotes the ratio of the fuel (H$_2$)
and oxidizer (O$_2$) at a given condition to the same quantity under stoichiometric
conditions. In this study, computations were performed at fuel-rich conditions,
$\Phi$~=~2.0. Homogeneous ignition at constant pressure is simulated using the
TChem software package~\cite{Safta:2011} using an initial pressure, $P_0$~=~1~atm and
initial temperature, $T_0$~=~900~K. The time required for the rate of
temperature increase to exceed a given threshold, regarded as \emph{ignition delay}
is recorded. 

We seek to understand the impact of uncertainty in the
rate-controlling parameters, pre-exponents~($A_i$'s) 
and the activation energies~($E_{a,i}$'s) 
as well as the initial pressure, temperature, and the
equivalence ratio on the ignition delay. The $\log(A_i)$'s associated with all
reactions and the $E_{a,i}$'s with non-zero nominal estimates
are considered to be uniformly distributed about their nominal estimates provided
in~\cite{Yetter:1991}. Temperature exponent, $n_i$ for each reaction is fixed to
its nominal value, also provided in~\cite{Yetter:1991}.
The initial conditions are also considered to be uniformly
distributed about their respective aforementioned values. 
The total number of uncertain inputs is 36 which makes
the present problem computationally challenging due to the large number of 
uncertain parameters in addition to the initial conditions.  
To address this challenge, we focus on reducing the dimensionality
of the problem by computing the active subspace.
This involves repeated evaluations of the gradient of a model output with
respect to the input parameters. Several numerical techniques are available
for computing the gradient, such as 
finite differences and more advanced methods involving
adjoints~\cite{Jameson:1988,Borzi:2011,Alexanderian:2017}. The
adjoint-based method requires a solution of the state equation (forward solve)
and the corresponding adjoint equation. Hence, it is 
limited by the availability of an adjoint solver. Additional model evaluations
at neighboring points are required if finite difference is used which increases
the computational effort. 
Regression-based techniques, which can be suitable for active subspace computations,
on the other hand, aim to estimate the gradient by approximating the model
output using a regression fit.  
These are computationally less intensive than
the former. However, as expected, there is a trade-off between computational
effort and accuracy in the two approaches for estimating the gradient. 

In this work, we adopt an iterative strategy to reduce the computational effort
associated with active subspace computation. Moreover, we explore two
approaches for estimating the gradient of the ignition delay with respect to
the uncertain rate-controlling parameters: pre-exponents ($A_i$'s), 
the activation energies ($E_{a,i}$'s),
as well as the initial conditions: $P_0$, $T_0$, and $\Phi_0$.
Note that the equivalence ratio corresponding to the initial molar ratios of
$H_2$ and O$_2$ is denoted as $\Phi_0$. The first approach uses finite
differences to estimate the gradient and will be referred to as the
\emph{perturbation} approach throughout the article.  The second approach
is adapted from~\cite[Algorithm 1.2]{Constantine:2015} and involves
repeated regression-fits to a subset of available model evaluations, and is
regarded as the regression approach in this work.

An alternate strategy to dimension reduction involves computing the global
sensitivity measures associated with the uncertain inputs of a model. Depending
upon the estimates of the sensitivity measures, only the important inputs are
varied for the purpose of uncertainty quantification (UQ). Sobol' indices are
commonly used as global sensitivity measures~\cite{Sobol:2001}. They are
used to quantify the relative contributions of the uncertain inputs to the variance
in the output, either individually, or in combination with other inputs. 
Multiple efforts have focused on efficient computation of the Sobol' 
indices~\cite{Sudret:2008,Plischke:2013,Tissot:2015,Li:2016} including the 
derivative-based global sensitivity measures~(DGSMs), developed to
compute approximate upper bounds for the Sobol' indices with much fewer
computations~\cite{Sobol:2009, Lamboni:2013}. It was noted
in~\cite{Diaz:2016,Constantine:2017} that DGSMs can be approximated by
exploiting their links with active subspaces. This led to the definition of the 
so-called \emph{activity scores}. In Section~\ref{sub:gsa}, we build on these
ideas to provide a complete analysis of links between Sobol indices, DGSMs, and
activity scores for functions of independent random inputs whose distribution
law belongs to a broad class of probability measures. 
It is worth mentioning that computing global sensitivity measures provides 
important information about a model that go beyond dimension reduction. By 
identifying parameters with significant impact on the model output, we can assess
regimes of validity of the model formulation, and gain critical insight into the
underlying physics in many cases. 

The main contributions of this paper are as follows: 
\begin{itemize}
\item 
Active subspace
discovery in a high-dimensional H$_2$/O$_2$ kinetics problem involving 36
uncertain inputs: The methodology presented in this work
successfully demonstrated that a 1-dimensional active subspace can reasonably 
approximate the uncertainty in the QoI, indicating immense potential
for computational savings. The presented
analysis can also guide practitioners in other problems of chemical kinetics on using the
method of active subspaces to achieve efficiency in uncertainty propagation.  
\item Comprehensive numerical investigation of the perturbation and the regression approaches: 
We investigate the suitability of both approaches
for estimating the gradient of ignition delay in the H$_2$/O$_2$ mechanism.
Specifically, we compare resulting
active subspaces, surrogate models, and the ability to approximate global
sensitivity measures through a comprehensive set of numerical experiments. Our
results reveal insight into the merits of the methods as well as
their shortcomings.  
\item Analysis of the links between  
global sensitivity measures: 
By connecting the recent theoretical advances in variance-based and
derivative-based global sensitivity analysis to active subspaces, we provide a complete analysis
of the links between total Sobol' indices, DGSMs, and activity scores for a broad
class of probability distributions. Our analysis is concluded by a result quantifying
approximation errors incurred due to fixing unimportant parameters, deemed so by 
computing their activity scores.  

\end{itemize}

This article is organized as follows. In section~\ref{sub:ac}, a brief
theoretical background on the active subspace methodology is provided. In
section~\ref{sub:gsa}, it is shown that the activity scores provide a
reasonable approximation to the DGSMs especially in a high-dimensional setting.
Additionally, a relationship between the three global sensitivity measures,
namely, the activity scores, DGSMs, and the total Sobol' indices is
established. In section~\ref{sec:method}, a systematic framework for computing
the active subspace is provided. 
Numerical results based on the perturbation approach are compared with those
obtained using the regression approach.  The active 
subspace is initially computed for a 19-dimensional 
H$_2$/O$_2$ reaction kinetics problem wherein only the $A_i$'s are
considered as uncertain. We further compute the active subspace
for a 36-dimensional H$_2$/O$_2$ reaction kinetics problem in section~\ref{sec:app}.
For both settings, the convergence characteristics and the predictive accuracy
of the two approaches is compared for a given amount of computational effort. The two
approaches are observed to yield consistent results, and a 1-dimensional active subspace
is observed to capture the uncertainty in the ignition delay.
Finally, a summary and discussion based on our findings is included in
section~\ref{sec:conc}.

\bigskip
\bigskip

\newcommand{\act}[3]{\nu_{{#2},{#3}}({#1})}
\newcommand{\actt}[3]{\tilde{\nu}_{{#2},{#3}}({#1})}
\newcommand{\redf}[2]{f^{({#1})}(\vec\xi; {#2})}
\section{Active subspaces}
\label{sub:ac}

Herein, we use a random vector 
$\vec\xi \in \Omega\in\mathbb{R}^{N_p}$ to parameterize model uncertainties, 
where $N_p$ is the number of uncertain inputs.
In practical computations, the \emph{canonical} variables $\xi_i, i=1,\ldots ,N_p$, are  
mapped to physical ranges meaningful in a given mathematical model. 
As mentioned in the introduction, an active subspace is a low-dimensional subspace
that consists of important directions in a model's input
parameter space~\cite{Constantine:2015}. The effective variability in a model output $f$
due to uncertain inputs is predominantly captured
along these directions. 
The directions constituting the active subspace are the dominant eigenvectors of the positive
semidefinite matrix 
\be
\mat{C} = \int_\Omega (\nabla_{\vec{\xi}}f)(\nabla_{\vec{\xi}}f)^\top \mu(d\vec\xi), 
\label{eq:C}
\ee
with 
$\mu(d\vec{\xi}) = \pi(\vec{\xi})d\vec{\xi}$, where $\pi(\vec{\xi})$ is the joint probability
density function of $\vec{\xi}$. Herein, $f$ is assumed to be a square integrable 
function with continuous partial 
derivatives with respect to the input parameters; moreover, we assume the partial derivatives
are square integrable. 
Since $\mat{C}$ is symmetric and
positive semidefinite, it admits a spectral decomposition:
\be
\mat{C} = \mat{W}\mat{\Lambda}\mat{W}^\top.
\ee
Here $\mat{\Lambda}$ = diag($\lambda_1,\ldots,\lambda_{N_p}$) with the eigenvalues
$\lambda_i$'s sorted in descending order
\[
     \lambda_1 \geq \lambda_2 \geq \cdots \geq \lambda_\Np \geq 0,
\] 
and $\mat{W}$ has the (orthonormal) eigenvectors $\vec{w}_1, \ldots, \vec{w}_\Np$ as its columns.
The eigenpairs are partitioned about the $r$th eigenvalue such that
$\lambda_r/\lambda_{r+1}\gg 1$, 
\be
 \mat{W} = [\mat{W}_1~\mat{W}_2],~~\mat{\Lambda} = \begin{bmatrix}\mat{\Lambda}_1 & \\  &
  \mat{\Lambda}_2. 
\end{bmatrix}
\ee
The columns of $\mat{W}_1 = 
\begin{bmatrix} \vec{w}_1 \cdots \vec{w}_r\end{bmatrix}$ 
span the dominant eigenspace of $\mat{C}$ and
define the active subspace, and $\mat{\Lambda}_1$ is a diagonal matrix with
the corresponding set of eigenvalues, $\lambda_1, \ldots, \lambda_r$, on its diagonal. 
Once the active subspace
is computed, dimension reduction is accomplished by transforming the parameter
vector $\vec\xi$ into 
$\vec{y} = \mat{W}_1^\top\vec{\xi} \in \R^r$. The elements of $\vec{y}$ are 
referred to as the set of active variables. 

Consider the function
\[
    G(\vec{y}) = f(\mat{W}_1\vec{y}), \quad \vec{y} \in \R^r.
\]
Following~\cite{Constantine:2015}, we use the approximation 
\[
f(\vec{\xi}) \approx f(\mat{W}_1 \mat{W}_1^T \vec\xi) =  
G(\mat{W}_1^\top \vec{\xi}).
\] 
That is, the model output $f(\vec\xi)$, in the original parameter space,
is approximated by $G(\mat{W}_1^\top \vec{\xi})$ in the active subspace.
We could confine uncertainty analysis to the inputs in the
active subspace whose dimension is typically much smaller (in applications that
admit such a subspace) than the dimension of the original input parameter. To further
expedite uncertainty analysis, one could fit a regression surface to $G$ using the 
following sequence of steps, as outlined in~\cite[chapter 4]{Constantine:2015}. 
\begin{enumerate}
\item Consider a given set of $N$ 
data points, $\big(\vec{\xi}_i, f(\vec{\xi}_i)\big)$, $i = 1, \ldots, N$. 
\item For each $\vec{\xi}_i$, compute $\vec{y}_i = \mat{W}_1^\top\vec{\xi}_i$. Note that
 $G(\vec{y}_i)$ $\approx$ $f(\vec{\xi}_i)$.
\item Use data points $\big(\vec{y}_i, f(\vec\xi_i)\big)$, $i = 1, \ldots, N$, to compute a 
regression surface $\hat{G}(\vec{y})\approx 
G(\vec{y})$.
\item Overall approximation, $f(\vec{\xi})$ $\approx$ $\hat{G}(\mat{W}_1^\top\vec{\xi})$.
\end{enumerate}

In practice, the matrix $\mat{C}$ defined in~\eqref{eq:C} is 
approximated using pseudo-random sampling techniques such as Monte Carlo or
Latin hypercube sampling (used in this work):
 %
 \be
 \mat{C}\approx \hat{\mat{C}} = \frac{1}{N}\sum\limits_{i=1}^{N} 
 (\nabla_{\vec{\xi}}f(\vec{\xi}_i))(\nabla_{\vec{\xi}}f(\vec{\xi}_i))^\top
 = \hat{\mat{W}}\hat{\mat{\Lambda}}\hat{\mat{W}}^\top
\label{eq:chat}
 \ee
Clearly the computational effort associated with constructing the matrix
$\hat{\mat{C}}$ scales with the number of samples, $N$. Hence, an iterative
computational approach is adopted in this work to gradually increase  
$N$ until the dominant eigenpairs are approximated
with sufficient accuracy; see Section~\ref{sec:method}. 


\section{GSA measures and their links with active subspaces}
\label{sub:gsa}
Consider a function $f = f(\xi_1, \xi_2, \ldots, \xi_\Np)$. 
While the active subspace framework described above does not make any assumptions
about independence of the inputs $\xi_i$, $i = 1, \ldots, \Np$, 
the classical 
framework of variance based sensitivity analysis~\cite{Sobol:2001, Saltelli:2010} 
assumes that the inputs
are statistically independent. While extensions to the cases 
of correlated inputs exist~\cite{Borgonovo:2007,Li:2010,Jacques:2006,Xu:2007},
 we limit the discussion in this section to the
case of random inputs that are statistically independent and are 
either uniformly distributed or
distributed according to the Boltzmann probability distribution.
Note that a measure $\mu$
on $\R$ is referred to as a Boltzmann measure if it is 
absolutely continuous with respect to the Lebesgue measure  
and admits a density  of the form $\pi(x) = C \exp\{-V(x)\}$,
where $V$ is a continuous function and $C$ a normalization 
constant~\cite{Lamboni:2013}. 
An important class of Boltzmann distributions are the so called log-concave
distributions, which include Normal, Exponential, Beta, Gamma, Gumbel, and
Weibull distributions. Note also that the uniform distribution does not fall under
the class of Boltzmanm distributions~\cite{Lamboni:2013}.

The total-effect Sobol' index ($T_i(f)$) of a model output, $f(\vec\xi)$ quantifies
the total contribution of the input, $\xi_i$ to the variance of the
output~\cite{Sobol:2001}. Mathematically, this can be expressed as follows:
\be
T_i(f) = 1 - 
\frac{\V_{\vec{\xi}_{\sim i}}\big[\mathbb{E}[f|\vec{\xi}_{\sim i}]\big]}{\V(f)},
\label{eq:total}
\ee
where $\vec{\xi}_{\sim i}$ is the input parameter vector with the  
$i^\text{th}$ entry removed. 
Here $\mathbb{E}[f|\vec{\xi}_{\sim i}]$ denotes the conditional
expectation of $f$ given $\vec{\xi}_{\sim i}$
and its variance is computed with respect to $\vec{\xi}_{\sim i}$.
The quantity, $\V(f)$ denotes the total variance of the model output.
 The total-effect Sobol' index accounts
for the contribution of a given input to the variability in the output by itself
as well as due to its interaction or coupling with other inputs. 
Determining accurate estimates of $T_i(f)$ typically involves a large 
number of 
model runs and is therefore can be prohibitive in the case of
compute-intensive applications. Derivative based 
global sensitivity measures (DGSMs)~\cite{Sobol:2009} provide a means for
approximating informative upper bounds on $T_i(f)$ at a lower cost; see 
also~\cite{Vohra:2018}. 

For $f: \Omega \to \R$, we consider the DGSMs,
\[
    \nu_i(f) := \E{\left(\frac{\partial f}{\partial\xi_i}\right)^2} =
                  \int_\Omega 
                  \left(\frac{\partial f}{\partial\xi_i}\right)^2
                  \pi(\vec{\xi})d\vec{\xi}, \quad i = 1, \ldots, \Np.   
\]
Here $\pi$ is the joint PDF of $\vec\xi$. Note that $\nu_i(f)$ is the 
$i^{\text{th}}$ diagonal element of the matrix $\mat{C}$ as defined in~\eqref{eq:C}. 
Consider the spectral decomposition written 
as $\mat{C} = \sum_{k=1}^\Np \lambda_k \vec{w}_k \vec{w}_k^\top$. Herein, we use the notation 
$\ip{\cdot}{\cdot}$ for the Euclidean inner product.
The following
result provides a representation of DGSMs in terms of the 
spectral representation of $\mat{C}$: 
\begin{lemma}
We have
$\nu_i(f) = \sum_{k=1}^\Np \lambda_k \ip{\vec{e}_i}{\vec{w}_k}^2$.
\end{lemma}
\begin{proof}
Note that $\nu_i(f) = \vec{e}_i^\top \mat{C} \vec{e}_i$,  
where $\vec{e}_i$ is the $i$th coordinate vector in $\R^\Np$, $i = 1, \ldots, \Np$.
Therefore,
$\nu_i(f) = \vec{e}_i^T \Big(\sum_{k=1}^\Np \lambda_k \vec{w}_k \vec{w}_k^\top\Big) \vec{e}_i
 = \sum_{k=1}^\Np \lambda_k \ip{\vec{e}_i}{\vec{w}_k}^2$. 
\end{proof}
In the case where the eigenvalues decay rapidly to zero, we can obtain
accurate approximations of $\nu_i(f)$ by truncating the summation: 
\[
   \act{f}{i}{r} =  \sum_{k=1}^r \lambda_k \ip{\vec{e}_i}{\vec{w}_k}^2,
   \quad i = 1, \ldots, \Np, \quad r \leq \Np.
\]
The quantities $\act{f}{i}{r}$ are called activity scores
in~\cite{Diaz:2016,Constantine:2017}, where links between GSA measures and
active subspaces is explored.
The following result, which
can also be found in~\cite{Diaz:2016,Constantine:2017}, quantifies the error in this
approximation. We provide a short proof for completeness. 
\begin{proposition}\label{prp:dgsm_bound} 
For $1 \leq r < \Np$,
\[
0 \leq \nu_i(f) - \act{f}{i}{r} \leq \lambda_{r+1}, \quad i = 1, \ldots, \Np.
\] 
\end{proposition}
\begin{proof} 


Note that, $\nu_i(f) - \act{f}{i}{r}= \sum_{k=r+1}^\Np \lambda_k \ip{\vec{e}_i}{\vec{w}_k}^2 \geq 0$,
which gives the first inequality. To see the upper bound, we note,
\[
   \sum_{k=r+1}^\Np \lambda_k \ip{\vec{e}_i}{\vec{w}_k}^2 \leq \lambda_{r+1} \sum_{k=r+1}^\Np \ip{\vec{e}_i}{\vec{w}_k}^2
   \leq \lambda_{r+1}. 
\]
The last inequality holds because 
$1 = \|\vec{e}_i\|_2^2 = 
\sum_{k = 1}^\Np \ip{\vec{e}_i}{\vec{w}_k}^2 
\geq \sum_{k=r+1}^\Np \ip{\vec{e}_i}{\vec{w}_k}^2$.
\end{proof} 
The utility of this result is realized in problems with 
high-dimensional parameters in which 
the eigenvalues $\lambda_i, i=1,\ldots,N_p$, decay rapidly to zero; in 
such cases, this result implies that  $\nu_i(f) \approx \act{f}{i}{r}$,
where $r$ is the \emph{numerical rank} of $\mat{C}$.  This will be especially
effective if there is a large gap in the eigenvalues.  

The relations recorded in the following lemma will be useful in the discussion 
that follows.
\begin{lemma}\label{lem:sum}
We have
\begin{enumerate}[label=(\alph*)]
\item $\sum_{i = 1}^\Np \act{f}{i}{r} = \sum_{k = 1}^r \lambda_k$. 
\item $\sum_{i = 1}^\Np \nu_i(f) = \sum_{k = 1}^\Np \lambda_k$. 
\end{enumerate}
\end{lemma}
\begin{proof}
The first statement of the lemma holds, because
\[
\sum_{i=1}^\Np \act{f}{i}{r}    
= \sum_{i=1}^\Np \sum_{k=1}^r \lambda_k \ip{\vec{e}_i}{\vec{w}_k}^2 
= \sum_{k=1}^r \lambda_k \sum_{i=1}^\Np \ip{\vec{e}_i}{\vec{w}_k}^2 
= \sum_{k=1}^r \lambda_k \| \vec{w}_k \|^2 = \sum_{k=1}^r \lambda_k.
\]
The statement (b) follows immediately from (a), because $\nu_i(f) = \act{f}{i}{\Np}$.
\end{proof}

It was shown in~\cite{Lamboni:2013} that the total-effect Sobol' 
index $T_i(f)$ can be bounded in terms of $\nu_i(f)$:
\begin{equation}\label{equ:sobol_bound}
T_i(f) \leq \frac{C_i}{\V(f)}\nu_i(f), \quad i = 1, \ldots, \Np,
\end{equation}
where for each $i$, $C_i$ is an appropriate \emph{Poincar\'{e}} constant
that depends on the distribution of $\xi_i$.
For instance, if $\xi_i$ is uniformly distributed on $[-1, 1]$, then $C_i = 4/\pi^2$; and in the 
case $\xi_i$ is normally distributed with variance $\sigma_i^2$, then $C_i = \sigma_i^2$. 
Note that~\eqref{equ:sobol_bound} for the special cases of 
uniformly distributed or normally distributed inputs was established first in~\cite{Sobol:2009}.
The bound~\eqref{equ:sobol_bound} provides a strong theoretical basis for using DGSMs to identify 
unimportant inputs. 

Combining Proposition~\ref{prp:dgsm_bound} and~\eqref{equ:sobol_bound}, shows
an interesting link between the activity scores and total-effect Sobol' indices.
Specifically, by computing the activity scores, we can identify the unimportant
inputs.  
Subsequently, 
one can attempt to reduce parameter dimension by fixing
unimportant inputs at nominal values. 

Suppose activity scores
are used to approximate DGSMs, and suppose
$\xi_i$ is
deemed unimportant as a result, due to a small activity score. 
We want to estimate
the approximation error that occurs once $\xi_i$ is fixed at a nominal value.
To formalize this process, we proceed as follows.
Let $\vec\xi$ be given and let $z$ be a nominal value for $\xi_i$.  
Consider the \emph{reduced} model, 
obtained by fixing $\xi_i$ at the nominal value: 
\[
\redf{i}{z} = f(\xi_1, \xi_2, \ldots, \xi_{i-1}, z, \xi_{i+1}, \ldots, \xi_\Np),
\] 
and consider the following relative error indicator:
\[
\mathcal{E}(z) =
\frac{ \int_\Omega \big( f(\vec\xi) - \redf{i}{z}\big)^2 \, \mu(d\vec\xi) }
          {\int_\Omega f(\vec\xi)^2 \, \mu(d\vec\xi)}.
\] 
This error indicator is a function of $z$ with $z$ distributed 
according to the distribution of $\xi_i$.
\begin{theorem}\label{thm:error_estimate}
We have $\Ez{ \mathcal{E}(z)} \leq 2C_i\big(\act{f}{i}{r} + \lambda_{r+1}\big)/{\V(f)}$, 
for $1 \leq r < \Np$.
\end{theorem}
\begin{proof} 
Note that, since 
$\int_\Omega f(\vec\xi)^2 \, \mu(d\vec\xi) = \V(f) + 
\left(\int_\Omega f(\vec\xi) \, \mu(d\vec\xi)\right)^2 \geq \V(f)$, we have
\[
\Ez{ \mathcal{E}(z)} \leq \frac{1}{\V(f)} \Ez{ 
\int_\Omega \big( f(\vec\xi) - \redf{i}{z}\big)^2 \, \mu(d\vec\xi)}
= 2 T_i(f), 
\]
where the equality can be shown using arguments similar to the proof of the main result 
in~\cite{SobolTarantolaGatelliEtAl07}. Using this, along with~\eqref{equ:sobol_bound} and
Proposition~\ref{prp:dgsm_bound}, we have 
\[
\Ez{ \mathcal{E}(z)} \leq 
\frac{2C_i}{\V(f)}\nu_i(f)
\leq 
\frac{2C_i}{\V(f)}\big[\act{f}{i}{r} + \lambda_{r+1}\big]. \qedhere
\]
\end{proof}

In~\cite{Vohra:2018} the screening metric
\be
   \tilde{\nu}_i(f) = \frac{C_i \nu_i(f)}{\sum_{i=1}^\Np C_i \nu_i(f)},
\label{eq:ndgsm}
\ee
was shown to be useful for detecting unimportant inputs. 
We can also bound the normalized DGSMs using activity scores as follows. It is straightforward to see  
that
\[
\tilde{\nu}_i(f) \leq 
\frac{ C_i \big(\act{f}{i}{r} + \lambda_{r+1}\big)}{\sum_{i=1}^\Np C_i \act{f}{i}{r}}
=\frac{C_i \act{f}{i}{r}}{\sum_{i=1}^\Np C_i \act{f}{i}{r}} + \kappa_i \lambda_{r+1}, 
\]
with $\kappa_i = C_i / (\sum_i C_i \act{f}{i}{r})$. 
In the case where where $\lambda_{r+1} \approx 0$, 
this motivates definition of
normalized activity scores
\[
   \actt{f}{i}{r} =  \frac{C_i \act{f}{i}{r}}{\sum_{i=1}^\Np C_i \act{f}{i}{r}}.
\] 

\begin{remark}
If the random inputs $\xi_i$, $i = 1, \ldots, \Np$, are iid, then 
the $C_i$'s in the definition of the normalized screening metric will cancel and 
\[
    \tilde{\nu}_i(f) = \frac{\nu_i(f)}{\sum_{i=1}^\Np \nu_i(f)} 
      = \frac{\sum_{k=1}^\Np \lambda_k \ip{\vec{e}_i}{\vec{w}_k}^2}{\sum_{k = 1}^\Np \lambda_k}.
\]
The expression for the denominator follows from Lemma~\ref{lem:sum}(b). 
Also, in the iid case, using Lemma~\ref{lem:sum}(a) we can simplify the normalized activity scores as follows. 
\[
   \actt{f}{i}{r} =  \frac{\act{f}{i}{r}}{\sum_{i=1}^\Np \act{f}{i}{r}} = 
                     \frac{\sum_{k=1}^r \lambda_k \ip{\vec{e}_i}{\vec{w}_k}^2}
                          {\sum_{k = 1}^r \lambda_k}.
\]
\end{remark}

The significance of the developments in this section are as follows.
Theorem~\ref{thm:error_estimate} provides a theoretical basis for parameter dimension
reduction using activity scores. This is done by providing an estimate of the
error between the reduced model and the original model. If a precise ranking of
parameter importance based on total-effect Sobol' indices is desired, one
can first identify unimportant inputs by computing activity scores, and then
perform a detailed variance based GSA of the remaining model parameters. This approach will 
provide great computational savings as variance based GSA will now be performed only
for a small number of
inputs deemed important based on their activity scores. Moreover, the
presented result covers a broad class of input distributions coming from 
the Boltzmann family of distributions.
Additionally, the normalized activity scores discussed above provide practical
screening metrics that require only computing the activity scores. This is in
contrast to the bound in Theorem~\ref{thm:error_estimate} that requires the
variance $\V(f)$ of the model output.

\bigskip
\bigskip

\section{Methodology}
\label{sec:method}

In this section, we outline the methodology for computing the
active subspace in an efficient manner. The proposed framework is employed to analyze 
a 19-dimensional H$_2$/O$_2$ reaction kinetics
problem whereby the logarithm of the 
pre-exponent ($A_i$) in the rate law associated with individual reactions provided in
Table~\ref{tab:kinetics} is considered to be uniformly
distributed in the interval, $[0.97\log(A_i^\ast), 1.03\log(A_i^\ast)]$;
$A_i^\ast$ is the nominal
estimate provided in~\cite{Yetter:1991}.
Two approaches are explored for estimating the gradient of ignition delay with
respect to $\log(A_i)$: 
a perturbation  approach that involves computation of model
gradients using finite difference in order to construct the matrix $\hat{\mat{C}}$
in~\eqref{eq:chat}, and a regression  approach 
that involves a linear regression fit to the available set of
model evaluations in order to approximate the gradient. 
The active subspace is computed in an iterative manner to avoid
unnecessary model evaluations once converged is established.


As discussed earlier, gradient estimation using finite differences 
requires additional model evaluations at the
neighboring points in the input domain. 
Hence, for $N$ samples in a $d$-dimensional parameter space, $N(d+1)$
model evaluations are needed. On the other hand, gradient estimation
using the regression-based approach involves a series of linear regression
fits to subsets of available evaluations as discussed 
in~\cite[Algorithm 1.2]{Constantine:2015}. Hence, the computational effort is  
reduced by a factor $(d+1)$ when using the regression-based approach.
In other words, for the same amount of computational effort, the regression
approach can afford a sample size that is (d+1) times larger than that in the
case of perturbation approach.
The specific sequence of steps
for computing the active subspace is discussed as follows.   

We begin by evaluating the gradient of the model output, $\nabla_{\bm{\xi}}f$,
at an initial set of $n_0$ samples (generated using Monte Carlo sampling)
denoted by $\bm{\xi}_i$, $i$ = $1,\ldots,n_0$.
Using the gradient
evaluations, the matrix, $\hat{\mat{C}}$ is computed. Eigenvalue decomposition
of $\hat{\mat{C}}$ yields an initial estimate of the dominant eigenspace,
$\hat{\mat{W}}_1$ and the set of corresponding eigenvalues, 
$\hat{\mat{\Lambda}}_1$. Note that $\hat{\mat{W}}_1$ is obtained by
partitioning the eigenspace around $\lambda_j$ such that the ratio of
subsequent eigenvalues,
$\left(\frac{\lambda_j}{\lambda_{j+1}}\right)\ge\mathcal{O}(10^1)$.
 At each subsequent iteration, model evaluations are
generated at a new set of $n_k$ samples. The new set of gradient evaluations
are augmented with the available set to re-construct $\hat{\mat{C}}$ followed
by its eigenvalue decomposition. The relative change in the norm of the
difference in squared value of individual components of the dominant
eigenvectors between subsequent iterations is evaluated. The process is
terminated and the resulting eigenspace is considered to have converged once
the maximum relative change at iteration $k$, $\max(\delta \hat{\mat{W}}_{1,j}^{(k)})$
($j$ is used as an index for the eigenvectors),
is smaller than a given tolerance, $\tau$.  A regression fit to
$G(\hat{\mat{W}}_1^\top\bm{\xi})$ is used as a surrogate to characterize and
quantify the uncertainty in the model output. Moreover, the components of the
eigenvectors in the active subspace are used to compute the activity scores, $\bm{\nu}_r(f)$,
which provide an insight into the relative importance of the uncertain inputs.
Note that the index, $r$, corresponds to the number of eigenvectors in
$\hat{\mat{W}}_1$. The sequence of steps as discussed are outlined
in Algorithm~\ref{alg:grad}.


\bigskip
\begin{breakablealgorithm}
\renewcommand{\algorithmicrequire}{\textbf{Input:}}
\renewcommand{\algorithmicensure}{\textbf{Output:}}
  \caption{An iterative strategy for discovering the active subspace}
  \begin{algorithmic}[1]
\Require $\theta_l$, $\theta_u$, $\beta$, $\tau$. 
\Ensure $\hat{\mat{\Lambda}}$, $\hat{\mat{W}}$, $\bm{\nu}_r(f)$ 
    \Procedure{Active Subspace Computation}{}
    \State Set $k$ = 0
	\State Draw $n_k$ random samples, $\{\bm{\xi}_i\}_{i=1}^{n_k}$ 
         according to $\pi_{\bm{\xi}}$. 
    \State Set $N_\text{total}$ = $n_k$ 
	\State For each $i=1, \ldots, N_\text{total}$, compute $f(\bm{\xi}_i)$ and the gradient $\bm{g}^i = \nabla_{\bm{\xi}}f(\bm{\xi}_i)$
	\State Compute $\hat{\mat{C}}$ and its eigenvalue decomposition 
		$\hat{\mat{C}}$= $\frac{1}{N_\text{total}}\sum\limits_{i=1}^{N_\text{total}}[\bm{g}^i][\bm{g}^i]^\top$ = 
		$\hat{\mat{W}}^{(k)}\hat{\mat{\Lambda}}^{(k)} \hat{\mat{W}}^{(k)\top}$
	\State Partition: $\hat{\mat{\Lambda}}^{(k)}=
        \begin{bmatrix} \hat{\mat{\Lambda}}_1^{(k)} & \\ & \hat{\mat{\Lambda}}_2^{(k)} \end{bmatrix}$, 
        $\hat{\mat{W}}^{(k)}=\begin{bmatrix} \hat{\mat{W}}_1^{(k)} & \hat{\mat{W}}_2^{(k)} \end{bmatrix}$, 
        $\hat{\mat{\Lambda}}_1^{(k)}\in \mathbb{R}^{N_p\times r}$
	\Loop
		\State Set $k$ = $k$ + 1
		\State Draw $n_k =  \lceil\beta n_{k-1}\rceil$  new random samples 
                $\{\bm{\xi}_i\}_{i=1}^{n_k}$  $\beta\in[0,1]$
                
		\State Set $N_\text{total}$ = $N_\text{total}$ + $n_k$ 
		\State Compute $\bm{g}^i = \nabla_{\bm{\xi}_i}f(\bm{\xi}_i)$, 
             	$i=n_{k-1}+1, \ldots, n_{k-1}+n_k$.  
		\State Compute $\hat{\mat{C}}$ = 
        	$\frac{1}{N_\text{total}}\sum\limits_{k=1}^{N_\text{total}}[\bm{g}^i][\bm{g}^i]^\top$
		\State Eigenvalue decomposition, $\hat{\mat{C}}$ = $\hat{\mat{W}}^{(k)}\hat{\mat{\Lambda}}^{(k)}
		 \hat{\mat{W}}^{(k)\top}$
		\State Partition the eigenspace of $\hat{\mat{C}}$ as shown in Step 7
		\State Compute $\delta \hat{\mat{W}}_{1,j}^{(k)}$ = 
                       \scalebox{1.25}{$\frac{\|(\hat{\mat{W}}_{1,j}^{k})^2 - 
                       (\hat{\mat{W}}_{1,j}^{k-1})^2\|_2}{\|(\hat{\mat{W}}_{1,j}^{k-1})^2\|_2}$}, 
                       $j = 1,\ldots,r$.
		\If {$\max\limits_{j}\left(\delta \hat{\mat{W}}_{1,j}^{(k)}\right)<\tau$}
			\State break
		\EndIf
	\EndLoop
	\State Compute $\nu_{i,r}(f) = \sum\limits_{j=1}^{r} \lambda_j w_{i,j}^2$,
	$i=1,\ldots,N_p$.
	\State Normalize $\nu_{i,r}(f)$ as $\tilde{\nu}_{i,r}(f)$ = \scalebox{1.25}{$\frac{\nu_{i,r}(f)}{\sum_i\nu_{i,r}(f)}$}.
	
    \EndProcedure
  \end{algorithmic}
  \label{alg:grad}
\end{breakablealgorithm}
\bigskip

To assess its feasibility and suitability, we implement
Algorithm~\ref{alg:grad} to compute the active subspace for 
the 19-dimensional H$_2$/O$_2$ reaction kinetics
problem by perturbing $\log(A_i)$ by 3$\%$ about its nominal
value as discussed earlier. For the purpose of verification,
$\hat{\mat{C}}$ was initially constructed using a large set of samples
($N$~=~1000) in the input domain. The gradient was estimated using
finite difference, and hence, a total of 20,000 model runs were performed. 
In Figure~\ref{fig:eig_comp}, we illustrate the comparison of
the resulting normalized eigenvalue spectrum by plotting 
$(\lambda_i/\lambda_0)$
($i = 1,\ldots,19$) corresponding to $N$~=~1000 and the same quantity corresponding to
a much smaller set of samples, $n$~=~$\{20,40,80,120\}$.
\begin{figure}[htbp]
 \begin{center}
  \includegraphics[width=0.45\textwidth]{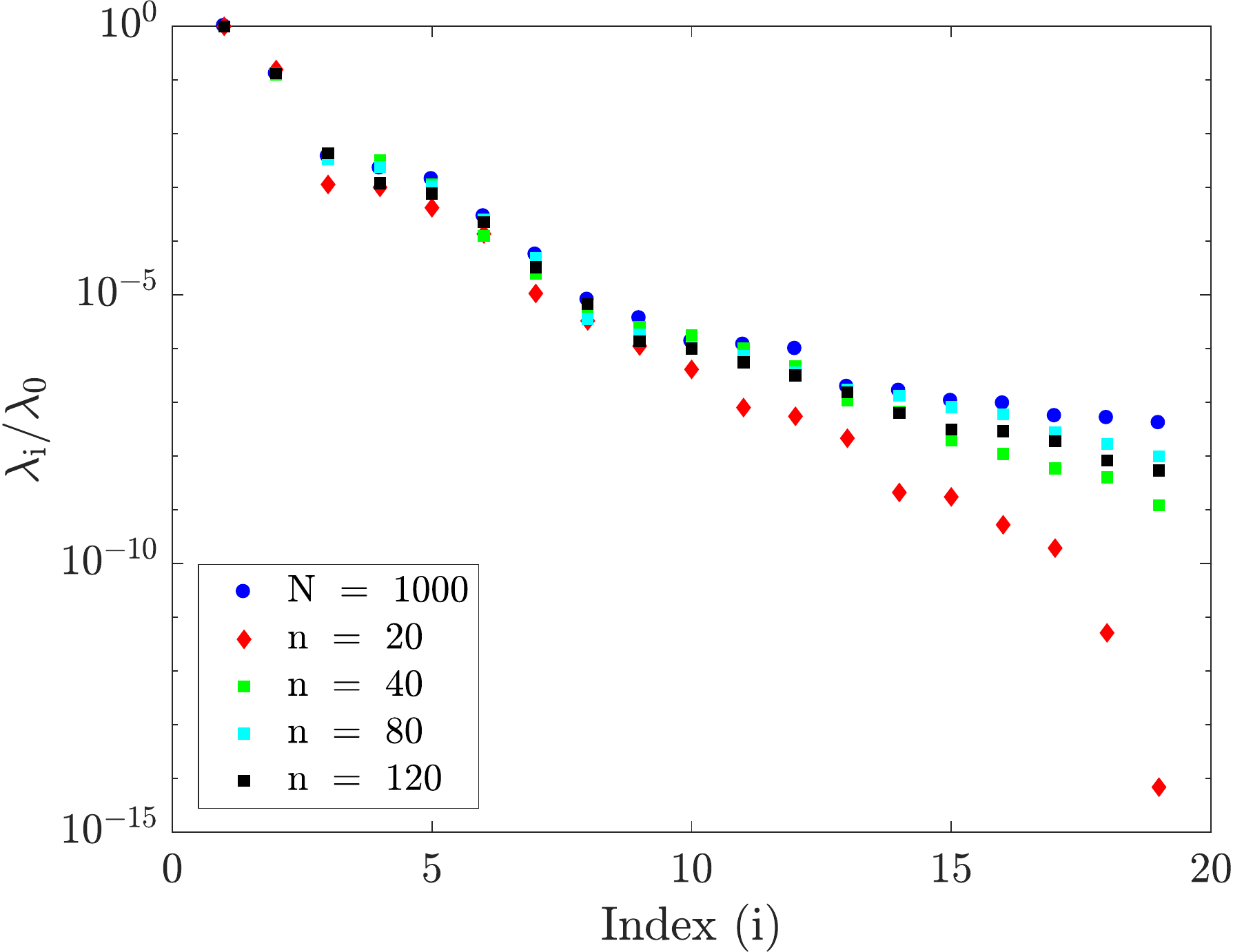}
\caption{A comparison of the normalized eigenvalue spectrum, 
($\lambda_i/\lambda_0)$
using $n$~=~$\{20,40,80,120\}$ samples with that
obtained using a much larger sample size, $N$~=~1000. 
} 
\label{fig:eig_comp}
\end{center}
\end{figure}
We observe that the dominant eigenvalues, $\lambda_1, \ldots, \lambda_4$, 
are approximated 
reasonably well with just 20 samples. As expected, the accuracy of higher-index
 eigenvalues is observed
to improve with the sample size. Since 
$\lambda_1$ is roughly an order of magnitude larger than $\lambda_2$, we expect 
a 1-dimensional active subspace to reasonably approximate the uncertainty
in the ignition delay. 
To further confirm this, we evaluate a relative L$^2$ norm of the difference
 ($\varepsilon_{\text{L}^{2}}^{N-n}$) between the 
squared value of corresponding components of the dominant eigenvector, computed using $N$~=~1000 ($\vec{w}_{1,N}$)
and $n$~=~$\{20,40,80,120\}$ ($\vec{w}_{1,n}$) as follows:
\be
\varepsilon_{\text{L}^{2}}^{N-n} = 
\frac{\|\vec{w}_{1,N}^2-\vec{w}_{1,n}^2\|_2}{\|\vec{w}_{1,N}^2\|_2}
\label{eq:accu}
\ee
The quantity, $\varepsilon_{\text{L}^{2}}^{N-n}$, was found to be 
$\mathcal{O}(10^{-2})$ in all cases.
Thus, even a small sample size, $n$ = 20, seems to approximate the dominant eigenspace with
reasonable accuracy in this case. 
The iterative strategy therefore offers a significant potential for computational gains. 


The active subspace for the 19-dimensional problem was also computed using regression-based
estimates of the gradient that do not require model evaluations at neighboring points as
discussed earlier.
The quantity, $\max\limits_j(\delta \hat{\mat{W}}_{1,j}^{(r)})$ defined in Algorithm~\ref{alg:grad} was
used to assess the convergence behavior of the two approaches. Using a set tolerance, $\tau$~=~0.05,
it was observed that both perturbation and regression approaches took 8 iterations to converge. 
Note that the computational effort at each iteration was considered to be the same in both cases.
More specifically, 5 new random samples were added for the perturbation approach at each iteration.
However, as discussed earlier, a total of 100 (=5$\times$(19+1)) model runs were needed to obtain the
model prediction and its gradients at these newly generated samples. Hence, in the case of
regression, 100 new random samples were generated at each iteration since gradient computation does
not require additional model runs in this case. Thus, including the initial step, a total of 900
model runs were required to obtain a converged active subspace in both cases. 

The accuracy of the two approaches was assessed by estimating $\varepsilon_{\text{L}^{2}}^{N-n}$ 
using the components of the dominant eigenvector in the converged active subspace in each case
in~\eqref{eq:accu}. The quantity, $\varepsilon_{\text{L}^{2}}^{N-n}$ was estimated to be 0.0657
and 0.1050 using perturbation and regression respectively. Hence, the regression approach was
found to be relatively more accurate. Squared values of the individual components of the 
dominant eigenvector from the
two approaches and for the case using $N$~=~1000 in the perturbation approach are plotted in
Figure~\ref{fig:comp}~(left).
\begin{figure}[htbp]
 \begin{center}
  \includegraphics[width=0.38\textwidth]{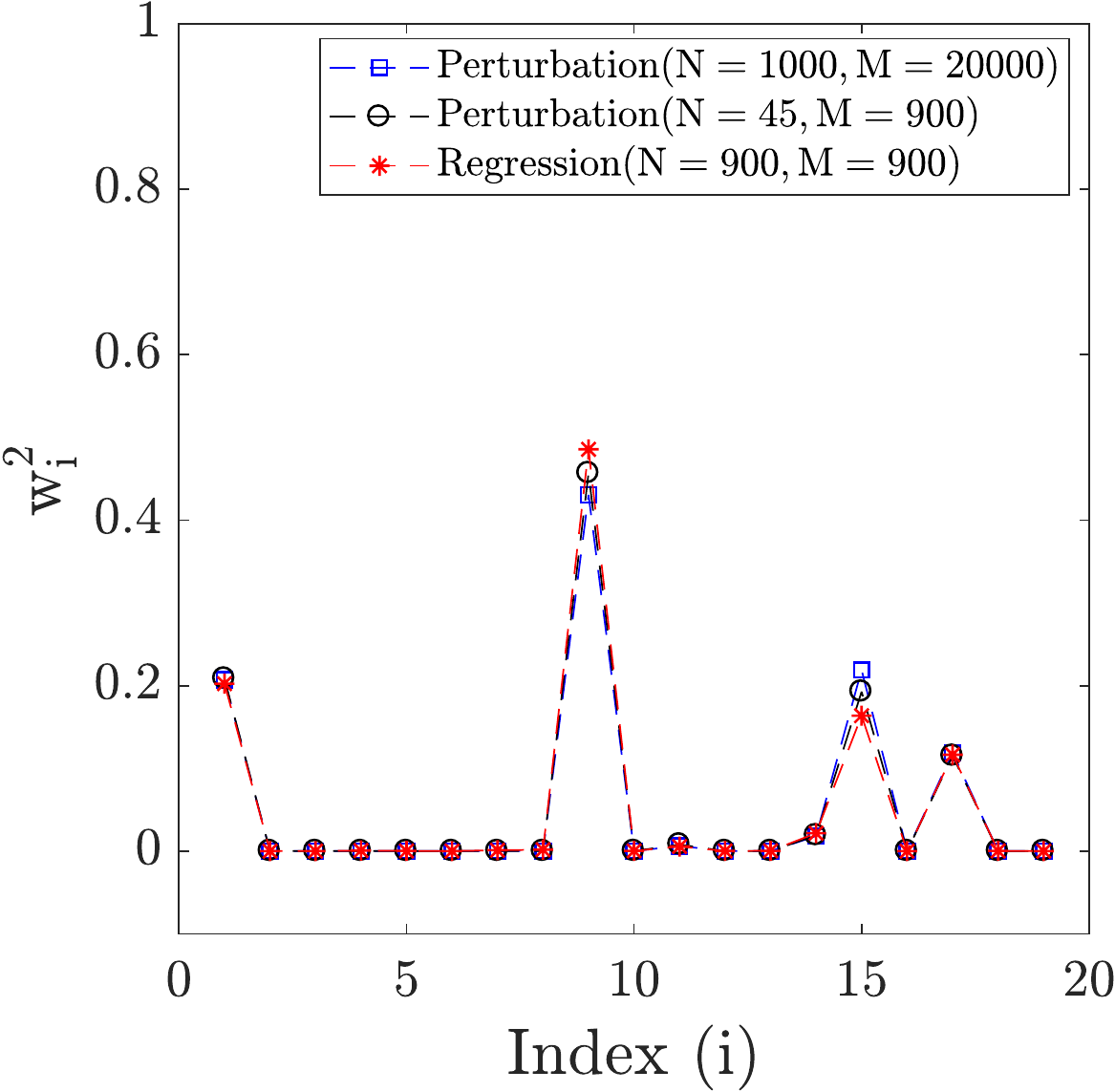}
  \includegraphics[width=0.45\textwidth]{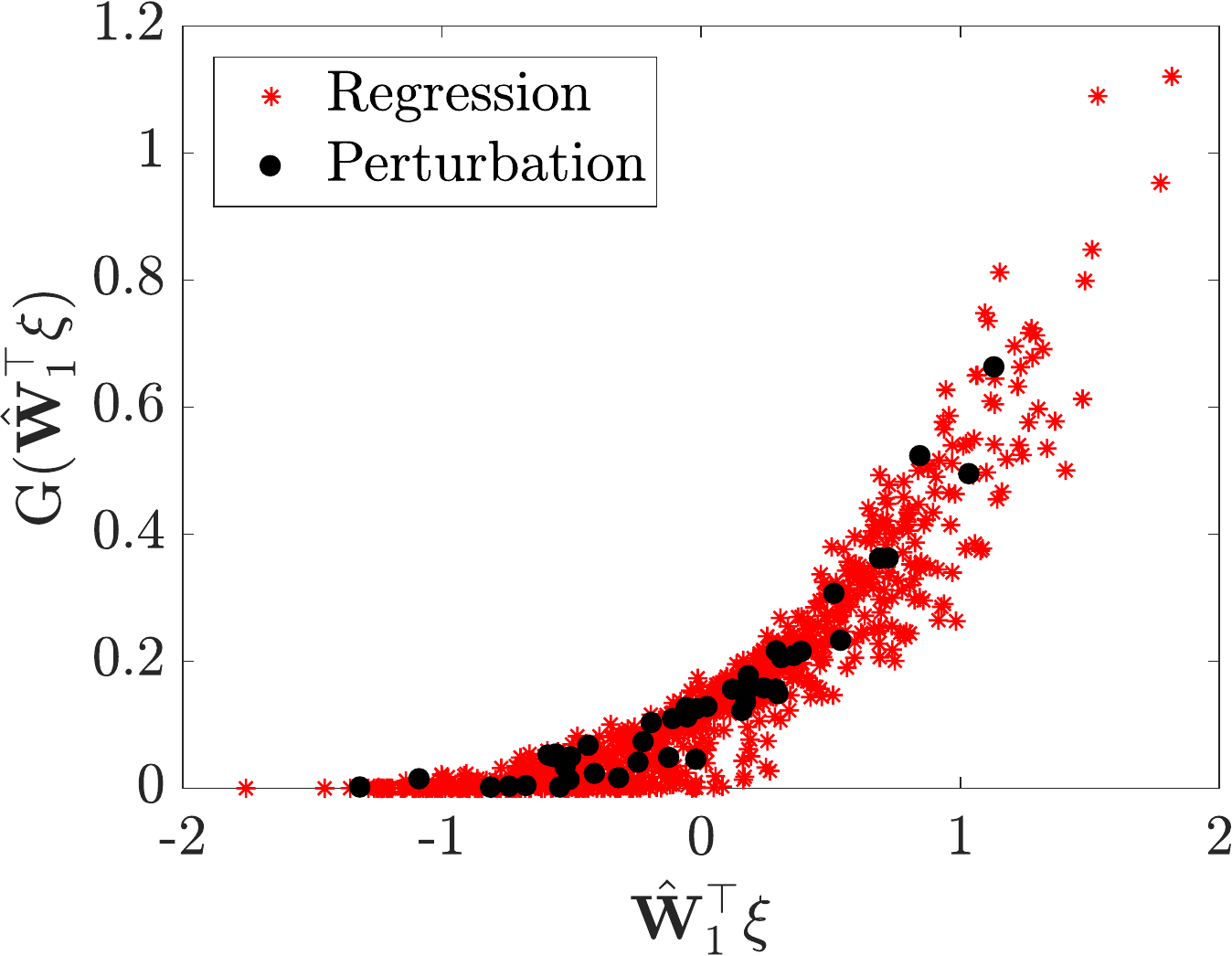}
\caption{Left: An illustrative comparison of individual squared components of the converged dominant
eigenvector obtained using perturbation and regression strategies using $M$~=~900 model runs in each case. 
Additionally, the dominant eigenvector components obtained using $M$~=~20000 model runs (corresponding to 
$N$~=~1000 samples)
in the perturbation strategy (test case), used to assess the accuracy of the two strategies are also plotted.
Right: An illustrative comparison of the SSPs generated using the
perturbation and the regression strategies for computing the active subspace. 
}
\label{fig:comp}
\end{center}
\end{figure}
The set of eigenvector components for the three cases are found to be in excellent agreement with
each other, indicating that both
approaches have sufficiently converged and are reasonably accurate for this setup.

As mentioned earlier, the model output $f(\bm{\xi})$ i.e. the ignition delay in the
H$_2$/O$_2$ reaction in this case, varies
 predominantly in a 1-dimensional active subspace. Hence, 
$f(\bm{\xi})$ can be approximated as $G(\hat{\mat{W}}_1^\top\bm{\xi})$ in
the 1-dimensional active subspace. The plot of $G$ versus $\hat{\mat{W}}_1^\top\bm{\xi}$, 
regarded as the \textit{sufficient summary plot} (SSP), obtained using the perturbation-based and
regression-based gradient estimates are
compared in Figure~\ref{fig:comp}~(right).
%
%
The dominant eigenvector obtained using perturbation is based on
$N$~=~45 samples which requires $M$~=~900 model runs. For the same amount of
computational effort, we can afford $N$~=~900 samples when using regression.
Hence, the SSP from regression is based on 900 points:
($\hat{\mat{W}}_1^\top\bm{\xi}_j$,~$G(\hat{\mat{W}}_1^\top\bm{\xi}_j)$),
$j = 1, \ldots, 900$. On the
other hand, the SSP from perturbation is plotted using only 45 points as
mentioned earlier.  Nevertheless, the illustrative comparison clearly indicates
that the two SSPs are in excellent agreement. Moreover, it is interesting to
note that the response in ignition delay based on the considered probability
distributions for $\log(A_i)$ although non-linear, can be approximated by a
1-dimensional active subspace.

We further estimate
the normalized activity scores for individual uncertain inputs ($\tilde{\nu}_{i,r}$; $r$=1 since a
1-dimensional active subspace seems reasonably accurate)
using the components of the dominant eigenvector as shown in Algorithm~\ref{alg:grad} (steps 21 and 22).
The activity scores for the 19 uncertain pre-exponents ($A_i$'s), estimated
using the perturbation and regression strategies are plotted in Figure~\ref{fig:comp_as}.  
\begin{figure}[htbp]
 \begin{center}
  \includegraphics[width=0.45\textwidth]{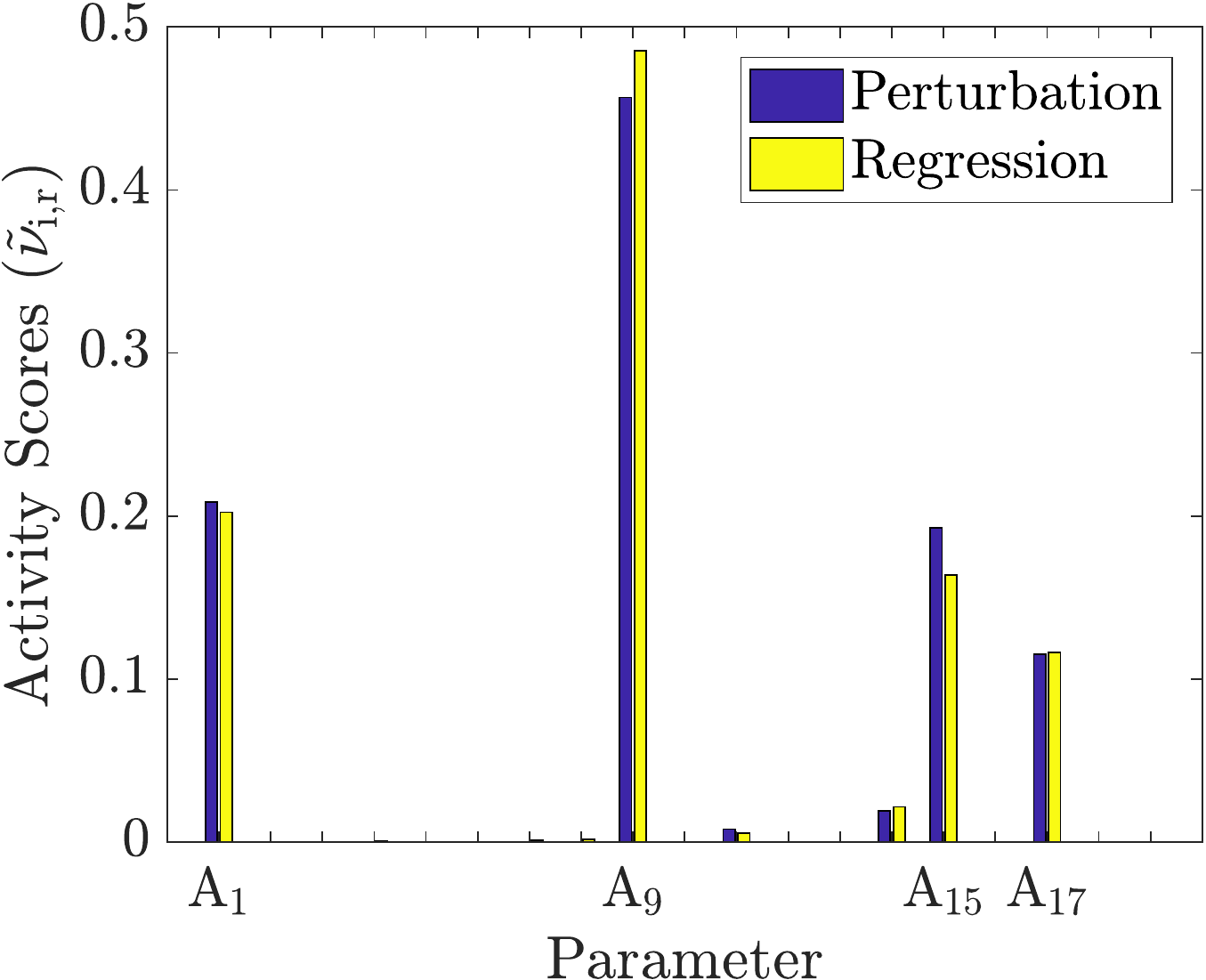}
\caption{Left: A bar-graph of normalized activity scores ($\tilde{\nu}_{i,r}$'s) 
for the 19 uncertain pre-exponents ($A_i$'s); $r$
denotes the number of eigenvectors in the dominant eigenspace.}
\label{fig:comp_as}
\end{center}
\end{figure}
The activity scores based on the two approaches for gradient estimation agree favorably 
with each other as well as those
based on the screening metric involving the DGSMs in~\cite{Vohra:2018}. It is observed that the uncertainty
associated with the ignition delay is largely due to the uncertainty in $A_9$ while $A_1$, $A_{15}$, and
$A_{17}$ are also observed to contribute significantly towards its variance.

The above comparisons indicate that the gradient of the ignition delay with respect to the
uncertain $A_i$'s is reasonably approximated using both, perturbation  and regression 
approaches in this case. Since both approaches yield consistent results and are
comparable in terms of convergence and accuracy, we could use either for the purpose
of active subspace computation for this setting.
In the following section, we shift our focus to 
the higher-dimensional H$_2$/O$_2$ reaction kinetics application wherein the activation
energies in the rate law as well as initial pressure, temperature, and stoichiometric conditions
are also considered to be uncertain.

\bigskip
\bigskip
\section{$\text{H}_2$/$\text{O}_2$ reaction kinetics: higher-dimensional case}
\label{sec:app}

For the high-dimensional case, we aim to investigate the impact of 
uncertainty in the following problem parameters on 
the ignition delay associated with the H$_2$/O$_2$ reaction:
(i) pre-exponents ($A_i$'s); (ii) the activation energies
($E_{a,i}$'s); and (iii) the initial pressure~($P_0$),
temperature~($T_0$), and stoichiometry~($\Phi_0$). 
The $\log(A_i)$'s, $E_{a,i}$'s for all reactions except $\mathcal{R}_6$
-- $\mathcal{R}_9$, $\mathcal{R}_{13}$ (due to zero nominal values for $E_a$),
and the initial conditions were considered to be uniformly distributed, and
perturbed by 2$\%$ about their nominal values.
Note that the magnitude of the perturbation was selected such that the
ignition delay assumes a physically meaningful value in the input domain. 
The nominal values of the rate parameters, $A_i$'s and $E_{a,i}$'s 
were taken from~\cite{Yetter:1991}. The nominal values of $P_0$, $T_0$, and
$\Phi_0$ were considered to be 1.0~atm, 900~K, and 2.0 respectively.

\subsection{Computing the active subspace}

The active subspace was computed using the iterative procedure outlined in 
Algorithm~\ref{alg:grad}. The convergence of the eigenvectors was examined
by tracking the quantity `$\max\limits_j(\delta \hat{\mat{W}}_{1,j}^{(i)})$'. 
In Figure~\ref{fig:conv_app}~(right), we examine $\max\limits_j(\delta \hat{\mat{W}}_{1,j}^{(i)})$
with increasing iterations for the perturbation and the regression approaches 
discussed earlier in Section~\ref{sec:method}. 
\begin{figure}[htbp]
 \begin{center}
  \includegraphics[width=0.8\textwidth]{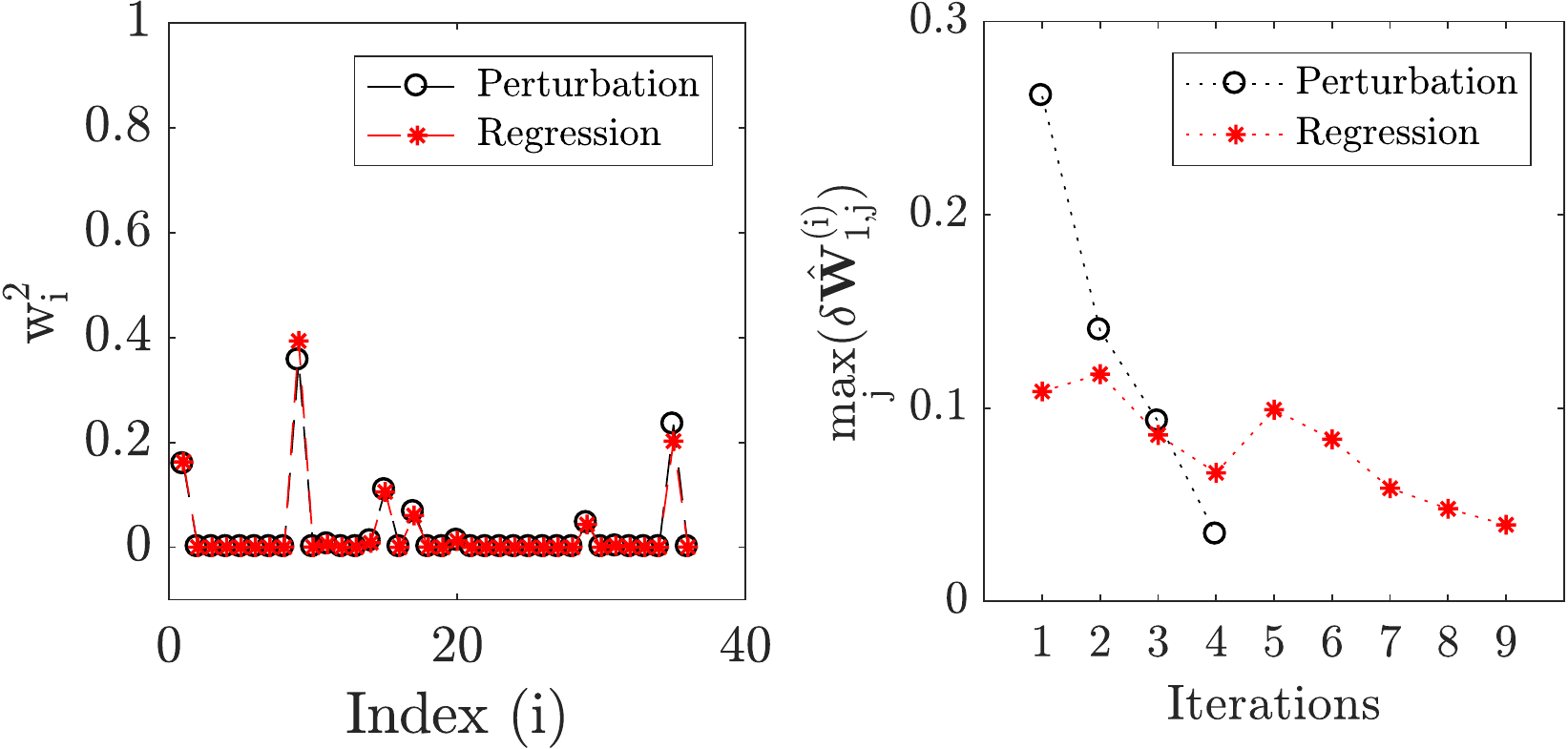}
\caption{Left: An illustrative comparison of individual components of the 
dominant eigenvector in the converged active subspace i.e., at the end
of 4 iterations in the perturbation approach and 9 iterations in the
regression approach. Right: A comparison of the convergence behavior
of the perturbation and the
regression approaches. Convergence is accomplished once 
$\max\limits_j(\delta \hat{\mat{W}}_{1,j}^{(i)})$ assumes a value smaller
than 0.05.}
\label{fig:conv_app}
\end{center}
\end{figure}
At each iteration, we improve our estimates of the matrix $\hat{\mat{C}}$ by
estimating the gradient of the ignition delay at 5 new randomly generated samples 
in the 36-dimensional input space. However, gradient computation at these
5 samples requires 185 (=5$\times$(36+1)) model runs when using perturbation. 
For the same computational effort, the regression approach can afford 185 new
samples at each iteration. It is observed that using $\tau$ = 0.05, the 
active subspace requires 4 iterations (925 model runs) to converge in the case of perturbation, and
9 iterations (1850 model runs) to converge in the case of regression. Hence, the computational effort
required to obtain a converged active subspace is doubled when using
regression to approximate the gradient. Moreover, gradient estimation in the perturbation approach
can be made more efficient by using techniques such as automatic differentiation~\cite{Kiparissides:2009}
and adjoint computation~\cite{Jameson:1988}. These techniques although not pursued here are 
promising directions for
future efforts pertaining to this work. In Figure~\ref{fig:conv_app}~(right), we compare
individual components of the dominant eigenvector in the converged active subspace
obtained using the two approaches. The components are observed to be in excellent
agreement with each other.

In Figure~\ref{fig:hd}, we plot the SSP for the perturbation approach (left) and the regression
approach (center) in a 1-dimensional active subspace. A 1-dimensional polynomial fit is also
illustrated in both cases. Moreover, the two surrogates are shown to be consistent with each other (right).
From these results, it is clear that a 1-dimensional active subspace captures the variability in the
ignition delay with reasonable accuracy, and that the two approaches yield consistent results.
\begin{figure}[htbp]
 \begin{center}
   \includegraphics[width=0.9\textwidth]{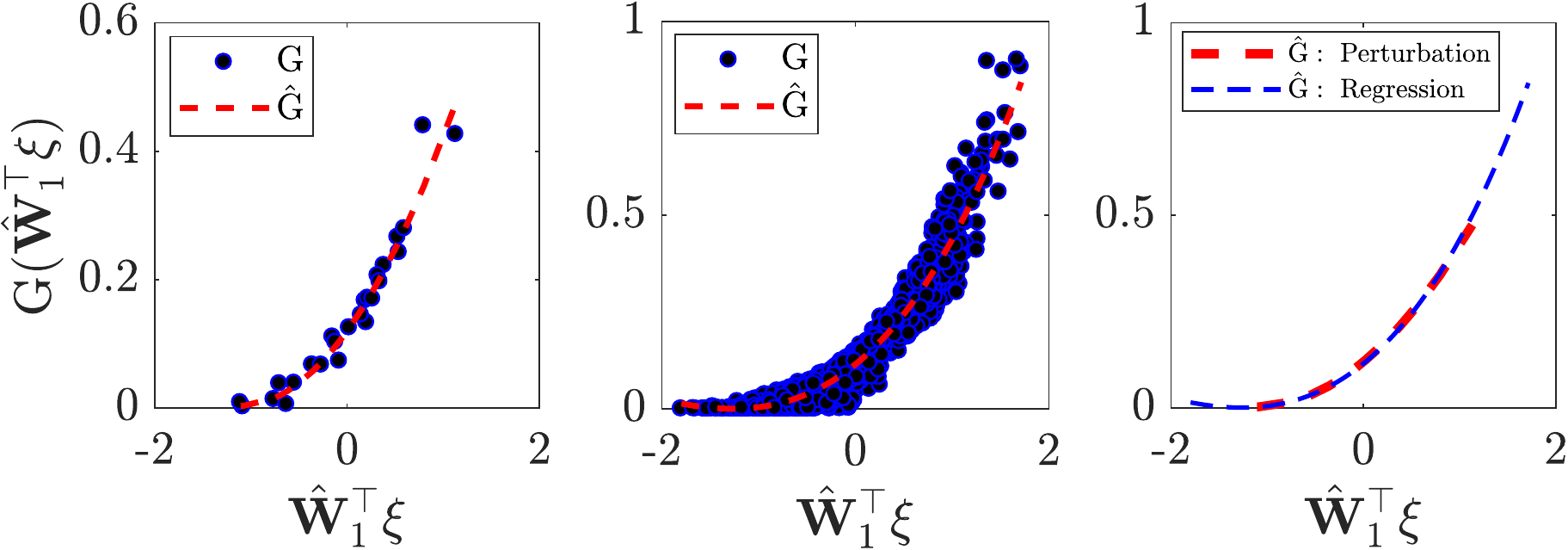}
\caption{Sufficient summary plots (SSPs) for the case of perturbation (left) and regression (center).
A polynomial fit of degree 2 and 3 as shown in the plots is used as a surrogate in the 
perturbation and regression approaches respectively. An illustrative comparison of the two
surrogates is also provided (right).}
\label{fig:hd}
\end{center}
\end{figure}

\subsection{Surrogate Assessment}
\label{sub:verify}

The 1-dimensional surrogate ($\hat{G}$) shown in Figure~\ref{fig:hd} for the perturbation
and regression approaches is investigated for its ability to capture the uncertainty in the
ignition delay. Specifically, we compare probability density functions (PDFs)
obtained using the true set of model evaluations, and 1-dimensional
surrogates ($\hat{G}$'s) based on the two approaches, as shown in 
Figure~\ref{fig:pdf_36D}. Note that the three PDFs were evaluated using the same set of 10$^4$ samples in the 
cross-validation set. 
\begin{figure}[htbp]
\begin{center}
\includegraphics[width=3.0in]{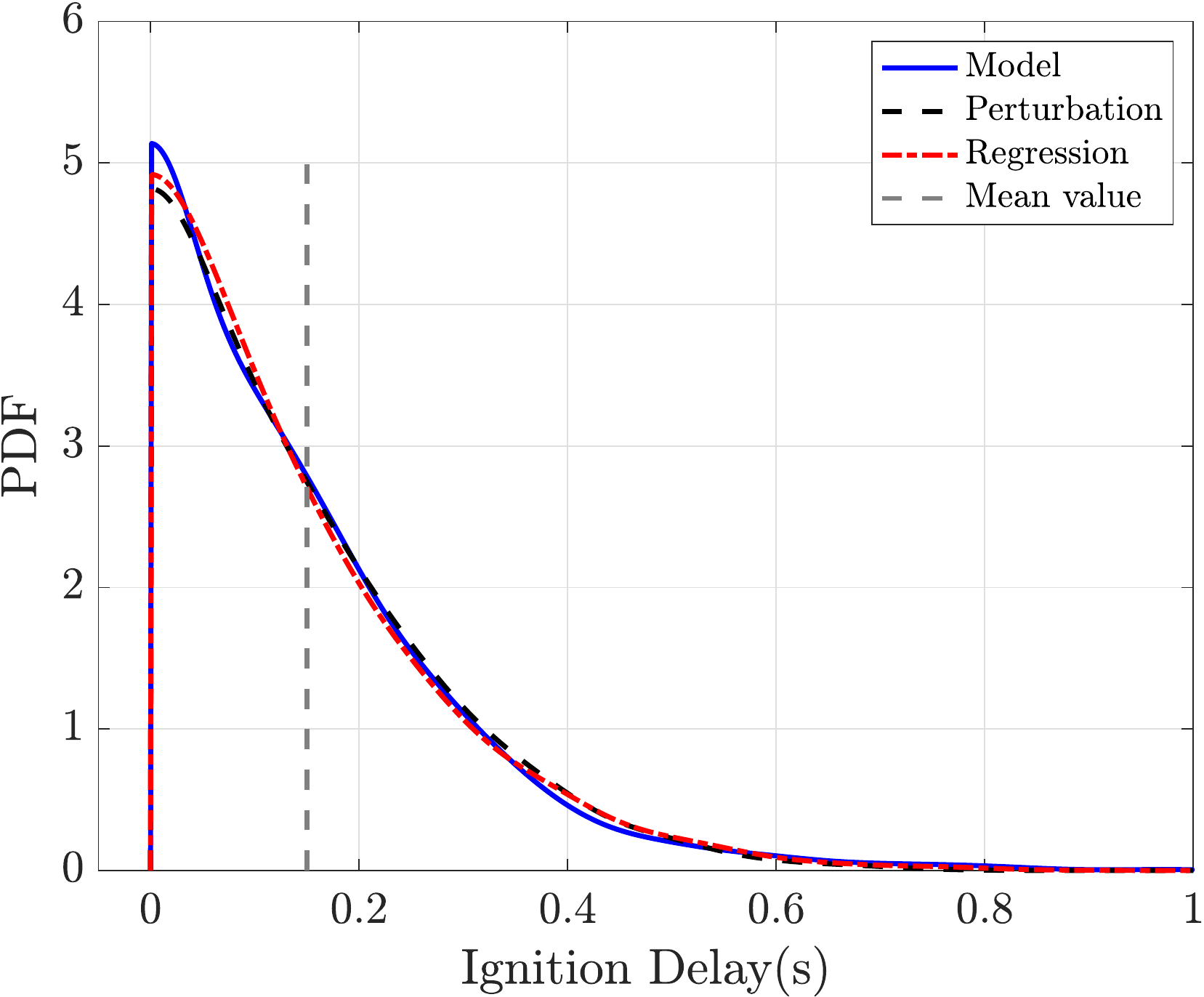}
\end{center} 
\caption{A comparison of the PDFs of ignition delay, obtained using model 
evaluations (solid line) and 1-dimensional surrogates using the regression-based strategy (dashed line) and the
perturbation-based strategy (dashed-dotted line). The same set of 10$^4$ samples in the cross-validation set were 
used in each case.}
\label{fig:pdf_36D}
\end{figure}
The PDFs are observed to be in close agreement with each other. Specifically, the modal
estimate and the uncertainty (quantified by the spread in the distributions) is found to be consistent for the three 
cases. To confirm this, we further compute the first-order (mean) and the second-order (standard deviation) 
statistics
of the estimates of the ignition delay obtained using the model, 1-dimensional surrogate from perturbation, and
1-dimensional surrogate from regression at the cross-validation sample set. The mean and standard deviation
estimates are provided in Table~\ref{tab:stats}.
\begin{table}[htbp]
\begin{center}
\begin{tabular}{ccc}
\toprule
$\textbf{Distribution}$ & $\mu$ & $\sigma$ \\ 
\bottomrule
$G$~(Model) & 0.15 & 0.14 \\
$\hat{G}$~(Perturbation-based) & 0.15 & 0.13 \\
$\hat{G}$~(Regression-based) & 0.15 & 0.13 \\
\bottomrule
\end{tabular}
\caption{The mean ($\mu$), and the standard deviation ($\sigma$), computed using the model ($G$), and
the surrogate ($\hat{G}$) based on the two strategies at 10$^4$ samples in the cross-validation
set.}
\label{tab:stats}
\end{center}
\end{table}
The mean and the standard deviation estimates obtained using the model and the 1-dimensional surrogates
are found to be in close agreement. Hence, the uncertainty in the ignition delay is accurately captured 
in both cases.

\subsection{GSA consistency check}

The normalized activity scores ($\tilde{\nu}_{i,r}$) based on the 1-dimensional active subspace, obtained
using the two approaches for gradient estimation (perturbation and regression),
are compared with the 
total-effect Sobol' indices
in Figure~\ref{fig:as_36D}. Note that the Sobol' indices were computed using the verified
1-dimensional surrogate ($\hat{G}$) in the active subspace, obtained using the 
perturbation approach. 
\begin{figure}[htbp]
 \begin{center}
  \includegraphics[width=0.8\textwidth]{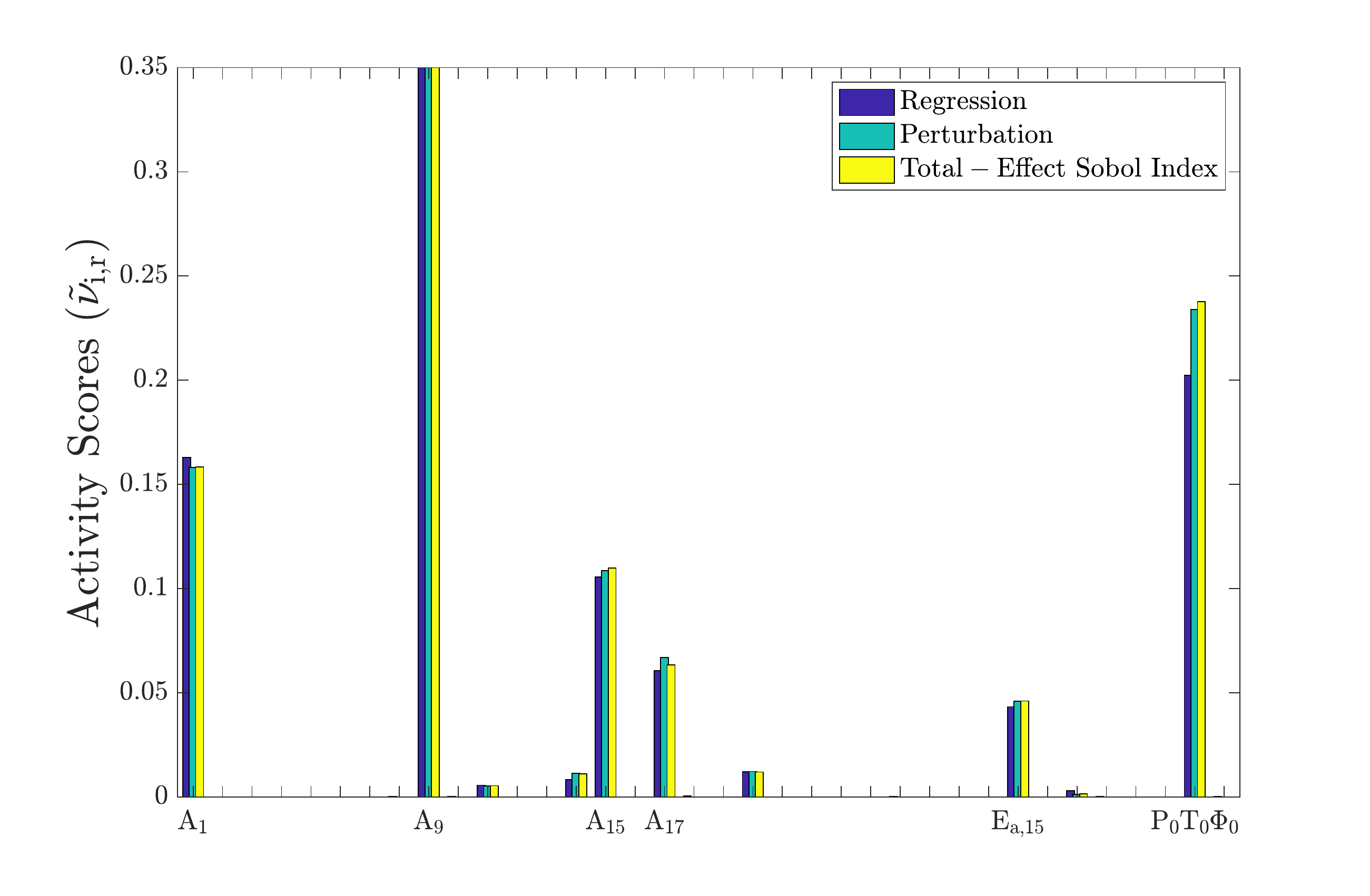}
\caption{Bar graphs illustrating individual activity scores for the uncertain 
rate parameters and the initial conditions for the H$_2$/O$_2$ reaction.}
\label{fig:as_36D}
\end{center}
\end{figure}
Several useful inferences can be drawn. Firstly, the normalized activity scores from the two approaches
and the total-effect Sobol' indices are found to be in close agreement with each other. Secondly, as expected, 
$\tilde{\nu}_{i,r}$ based on perturbation exhibits a better agreement with the total-effect Sobol' indices since
the 1-dimensional
surrogate based on the same approach was used to evaluate the Sobol' indices. This observation demonstrates
that the proposed framework is self-consistent. Thirdly, the variability in the ignition delay is predominantly due
to the uncertainty in $A_1$, $A_9$, and $T_0$ while contributions from the uncertainty in $A_{15}$, $A_{17}$, and
$E_{a,15}$, and $T_0$ are also found to be significant. The remaining rate parameters, initial pressure ($P_0$),
and  stoichiometry ($\Phi_0$) do not seem to impact the ignition delay in their considered intervals. 
Therefore, GSA has helped
identify the important rate parameters i.e. key contributors to the uncertainty, and also demonstrated that among the 
considered uncertain initial conditions, the ignition delay is mainly impacted by the perturbations in the initial 
temperature in the considered interval.

As shown in the PDF plotted in Figure~\ref{fig:pdf_36D}, the ignition delay assumes a wide range of values, i.e.
from about 2~ms to 400~ms. However, for many practical applications, a much smaller ignition delay (0.1~ms--1~ms) might
be of interest. The authors would like to point out that the proposed framework was also implemented to such a regime
by using a nominal value of the initial temperature, $T_0$ = 1000~K, and the initial pressure, $P_0$ = 1.5~atm. 
Our analysis for this
regime once again revealed that a 1-dimensional active subspace was able to capture the variability in the ignition delay due to
the uncertainty in the rate-controlling parameters and the input conditions. The sensitivity trends were also found to
be qualitatively similar to those presented in Figure~\ref{fig:as_36D}. We have not included these results in the interest
of brevity. Therefore, the proposed methodology was tested
for is robustness and applicability for a wide range of conditions pertaining to the considered application. 

\bigskip
\bigskip

\section{Summary and Discussion}
\label{sec:conc}
 
In this work, we focused on the uncertainty associated with the
rate-controlling parameters in the H$_2$/O$_2$ reaction mechanism
as well as the initial pressure, temperature, and stoichiometry and its
impact on ignition delay predictions. The mechanism involves 19 different
reactions and in each case, the reaction rate depends upon the choice of a
pre-exponent and an activation energy. Hence, in theory, the evolution of the
chemical system depends upon 38 rate parameters and three initial conditions.
However, we considered 
epistemic uncertainty in all pre-exponents and activation energies with non-zero
nominal values i.e. a total of 33 rate parameters instead of 38 in addition to
the three initial conditions.  
To facilitate efficient uncertainty analysis, we focused our efforts on
reducing the dimensionality of the problem by identifying important directions
in the parameter space such that the model output 
predominantly varies along these directions. These important directions
constitute the active subspace. Additionally, we demonstrated that the activity scores,
computed using the components of the dominant eigenvectors provide an efficient
means for approximating derivative based global sensitivity measures (DGSMs).
Furthermore, we established generalized mathematical linkages between the
different global sensitivity measures: activity scores, DGSMs, and total Sobol'
index which could be exploited to reduce computational effort associated with
global sensitivity analysis. 
 
Active subspace computation requires repeated evaluations of the gradient of
the QoI i.e. the ignition delay. For this purpose, we explored two approaches,
namely, perturbation and regression. Both approaches were shown to yield
consistent results for the 19-dimensional problem wherein only the
pre-exponents were considered to be uncertain. It was observed that the
computational effort required to obtain a converged active subspace was
comparable for the two approaches. However, the predictive accuracy of the
perturbation approach was found to be relatively higher. Moreover, a
1-dimensional active subspace was shown to reasonably approximate the
uncertainty in the ignition delay. Additionally, the activity scores were also
shown to be consistent with the screening metric estimates based on DGSMs
in~\cite{Vohra:2018}. An iterative procedure was adopted to enhance the
computational efficiency. 

The active subspace was further computed for a 36-dimensional problem
wherein all pre-exponents and activation energies with non-zero nominal
estimates as well as the initial conditions were considered uncertain. 
Once again, consistent results were obtained using the two approaches.
A 1-dimensional active subspace was shown to reasonably
capture the uncertainty in the ignition delay in this case. However, the
computational effort required to compute a converged active subspace
using perturbation was found to be half of the effort required in the case
of regression. Predictive accuracy of the two approaches was found to 
be comparable. Hence, perturbation seems like a preferred approach
for the higher-dimensional problem based on our findings. GSA results indicated
that the variability in the ignition delay is predominantly due to the 
uncertainty in the rate parameters, $A_1$ and $A_9$ with significant
contributions from $A_{15}$, $A_{17}$, and $E_{a,15}$. Additionally, the
ignition delay was found to be sensitive towards $T_0$.

Based on our findings, the perturbation approach is preferable for active
subspace computation; the computational cost of this approach can be reduced
significantly, if more efficient gradient computation techniques (e.g.,
adjoint-based approaches or automatic differentiation) are feasible. The
regression-based approach can be explored in situations involving intensive
simulations where gradient computation is very challenging. 

%
%
We also mention that alternate regression-based approaches such as ones based
on computing a global quadratic model have been proposed and used in the
literature; see e.g.,~\cite{Constantine:2017a}.  The applicability of such an
approach in the context of high-dimensional chemical reaction networks is
subject to future work. 

The computational framework presented in this work is agnostic to the choice of
the chemical system and can be easily adapted for other systems as long as the
quantity of interest is continuously differentiable in the considered domain of
the inputs.  We have demonstrated that the active subspace could be exploited
for efficient forward propagation of the uncertainty from inputs to the output.
The resulting activity scores and the low-dimensional surrogate could further
guide optimal allocation of computational resources for calibration of the
important rate-controlling parameters and input conditions in a
Bayesian setting.  Additionally, dimension reduction using active subspaces
could assist in developing robust formulations for predicting discrepancy
between simulations and measurements due to epistemic uncertainty in the model
inputs.

\bigskip
\bigskip

\section*{Acknowledgment}

M.~Vohra and S.~Mahadevan gratefully acknowledge funding support from the
National Science Foundation (Grant No. 1404823, CDSE Program), and Sandia
National Laboratories (PO No. 1643376, Technical monitor: Dr. Joshua Mullins).
The research of A.~Alexanderian was partially supported by the
National Science Foundation through the grant DMS-1745654.  The authors also
thank Dr.~Cosmin Safta at Sandia National Laboratories for his guidance
pertaining to the usage of the TChem software package.

\bigskip
\bigskip

\bibliographystyle{elsarticle-num}
\bibliography{REFER}

\end{document}